\documentclass[british,english]{article}
\usepackage[T1]{fontenc}
\usepackage[latin9]{inputenc}
\usepackage[paperwidth=16cm,paperheight=25cm]{geometry}
\usepackage{color}
\usepackage{babel}
\usepackage{amsthm}
\usepackage{amsmath}
\usepackage{amssymb}
\usepackage{stmaryrd}
\usepackage{esint}
\usepackage[unicode=true,pdfusetitle,
 bookmarks=true,bookmarksnumbered=false,bookmarksopen=false,
 breaklinks=false,pdfborder={0 0 1},backref=false,colorlinks=true]
 {hyperref}
\usepackage{breakurl}

\makeatletter
\newcommand{\lyxaddress}[1]{
\par {\raggedright #1
\vspace{1.4em}
\noindent\par}
}
\theoremstyle{plain}
\newtheorem{thm}{\protect\theoremname}
  \theoremstyle{plain}
  \newtheorem{prop}[thm]{\protect\propositionname}
  \theoremstyle{plain}
  \newtheorem{lem}[thm]{\protect\lemmaname}
  \theoremstyle{plain}
  \newtheorem{cor}[thm]{\protect\corollaryname}
  \theoremstyle{remark}
  \newtheorem{rem}[thm]{\protect\remarkname}

\makeatother

  \addto\captionsbritish{\renewcommand{\corollaryname}{Corollary}}
  \addto\captionsbritish{\renewcommand{\lemmaname}{Lemma}}
  \addto\captionsbritish{\renewcommand{\propositionname}{Proposition}}
  \addto\captionsbritish{\renewcommand{\remarkname}{Remark}}
  \addto\captionsbritish{\renewcommand{\theoremname}{Theorem}}
  \addto\captionsenglish{\renewcommand{\corollaryname}{Corollary}}
  \addto\captionsenglish{\renewcommand{\lemmaname}{Lemma}}
  \addto\captionsenglish{\renewcommand{\propositionname}{Proposition}}
  \addto\captionsenglish{\renewcommand{\remarkname}{Remark}}
  \addto\captionsenglish{\renewcommand{\theoremname}{Theorem}}
  \providecommand{\corollaryname}{Corollary}
  \providecommand{\lemmaname}{Lemma}
  \providecommand{\propositionname}{Proposition}
  \providecommand{\remarkname}{Remark}
\providecommand{\theoremname}{Theorem}

\begin{document}

\title{A note on one of the Markov chain Monte Carlo novice's questions.}

\author{Christophe Andrieu}

\maketitle

\lyxaddress{School of Mathematics, University of Bristol, UK.}
\begin{abstract}
We introduce a novel time-homogeneous Markov embedding of a class
of time inhomogeneous Markov chains widely used in the context of
Monte Carlo sampling algorithms which allows us to answer one of the
most basic, yet hard, question about the practical implementation
of these techniques. We also show that this embedding sheds some light
on the recent result of \cite{maire-douc-olsson2013}. We discuss
further applications of the technique.
\end{abstract}
\textbf{\textit{\textcolor{black}{Keywords}}}: Markov chain Monte
Carlo, Metropolis within Gibbs, Peskun order, deterministic scan,
random scan.

\section{Introduction}

It is often said that little time is needed for a novice Markov chain
Monte Carlo user to embarrass an expert with apparently simple questions.
One such question is the following. Let $\pi$ be a probability distribution
defined on some measurable space $\bigl(\mathsf{X},\mathcal{X}\bigr)$
and for some $k\in\mathbb{N}^{*}$ let $\mathfrak{P}:=\left\{ \Pi_{i}:\mathsf{X}\times\mathcal{X}\rightarrow[0,1],i=1,\ldots,k\right\} $
be a family of Markov transition kernels assumed to be reversible
with respect to $\pi$. Markov chain Monte Carlo methods consist of
using these Markov transitions in order to simulate realizations of
a Markovian process $\left\{ X_{i},i\geq0\right\} $ which may be
used to approximate expectations of functions $f:\mathsf{X}\rightarrow\mathbb{R}$
with respect to $\pi$, $\mathbb{E}_{\pi}\left(f(X)\right)$ for $X\sim\pi$,
with estimators of the form 
\[
S_{M}\big(f\big)=\frac{1}{M}\sum_{i=0}^{M-1}f\bigl(X_{i}\bigr)\;.
\]
A natural question is how to best use $\mathfrak{P}$ in order to
minimize the variability of this estimator? Traditionally there are
essentially two approaches to construct such Markov chains from $\mathfrak{P}$.
The first one consists of considering the homogeneous Markov chain
with transition defined as a mixture of the transitions in $\mathfrak{P}$
\begin{equation}
P^{{\rm rand}}=\frac{1}{k}\sum_{j=1}^{k}\Pi_{j}\;,\label{eq:randomscan}
\end{equation}
which corresponds to choosing one of the kernels at random at each
iteration. The second option consists of cycling through $\mathfrak{P}$
in a deterministic fashion, which defines a non-homogeneous Markov
chain. More precisely, for $k\in\mathbb{N}_{*}$ define the forward
circular permutation $\sigma:\left\{ 1,\ldots,k\right\} \rightarrow\left\{ 1,\ldots,k\right\} $
such that $\sigma\bigl(j\bigr)=j+1$ for $j\in\left\{ 1,\ldots,k-1\right\} $
and $\sigma\bigl(k\bigr)=1$, and its powers, starting with $\sigma^{0}(j)=j$
for $j\in\{1,\ldots,k\}$ and $\sigma^{i}=\sigma^{i-1}\circ\sigma$
for $i\geq1$. Then we introduce the sequence of Markov transition
probabilities $P^{{\rm strat}}=\left\{ P_{i}^{{\rm strat}}=\Pi_{\sigma^{i-1}(1)},i\geq1\right\} $.
For simplicity in the two scenarios outlined above we use the same
notation $\mathbb{P}_{\mu}(\cdot)$ (resp. $\mathbb{E}_{\mu}(\cdot)$)
for the probability distribution (resp. expectation) of the two Markov
chains such that $X_{0}\sim\mu$ where $\mu$ is a probability distribution
on $\big(\mathsf{X},\mathcal{X}\big)$. For example, for $i\geq1$,
with $\mathcal{F}_{i}:=\sigma\big(X_{0},X_{1},\ldots,X_{i}\big)$
and $A\in\mathcal{X}$

\[
\mathbb{P}_{\mu}\bigl(X_{i}\in A\mid\mathcal{F}_{i-1}\bigr)=\begin{cases}
P^{{\rm rand}}\big(X_{i-1},A\big)\\
P_{i}^{{\rm strat}}\big(X_{i-1},A\big)
\end{cases}\quad,
\]
depending on the scenario considered. There shall not be room for
confusion in what follows. Similarly, letting ${\rm var_{\mu}(\cdot)}$
be the variance operator corresponding to $\mathbb{E}_{\mu}(\cdot)$,
we define for $P^{\star}=\big\{ P^{{\rm rand}}\big\}$ or $P^{\star}=P^{{\rm strat}}$
and $f:\mathsf{X}\rightarrow\mathbb{R}$, the asymptotic variance
\[
\mathrm{var}(f,P^{\star})=\lim_{M\rightarrow\infty}\mathrm{var}_{\pi}\left(M^{1/2}S_{M}\big(f\big)\right)\;,
\]
when the limit exists. The novice's question we are interested in
here is which of the two schemes one should use in order to minimize
the asymptotic variance of the estimator $S_{M}\big(f\big)$? Despite
a long interest in ordering Markov chain Monte Carlo methods in terms
of such performance measure \cite{HASTINGS01041970,peskun,caracciolo-pelissetto-sokal,tierney-note,maire-douc-olsson2013},
the novice's question is, to the best of our knowledge, still unanswered
and surprisingly hard; we provide here only a partial answer as we
show that for $k=2$ it is always preferable to use the deterministic
update; see Theorem \ref{thm:mainresult} for a precise statement.
We note that if the measure of performance is the time for convergence
to equilibrium then no general conclusions can be made, see \cite{roberts1997updating}
and the references therein. Indeed, in tractable scenarios involving
the so-called Gibbs sampler, it can be established that neither of
the schemes dominates the other uniformly and that the dependence
structure of the targeted distribution $\pi$ determines this ordering.

The main idea of our proof consists of embedding the inhomogeneous
Markov chain induced by the sequence $P^{{\rm strat}}$ into an homogeneous
Markov chain, of transition $T$ defined on an extended space, and
rewriting the asymptotic variance of the inhomogeneous chain in terms
of the resolvent of the operator $T$ corresponding to this embedding
chain (Eq. (\ref{eq:generalisedresolventexpressionvariance})). This
leads to a generalization of a well known and simple identity for
the asymptotic variance of homogeneous Markov chains. It turns out
that in the case $k=2$ the homogeneous Markov chain defined by $P^{{\rm rand}}$
can be seen as being the self-adjoint part of the operator $T$. This
together with another variational representation of the asymptotic
variance of Markov chains in terms of their self-adjoint and skew
symmetric parts allows us to answer the novice's question when $k=2$
(Theorem \ref{thm:mainresult}). Along the way we also show that our
approach provides us with a very short proof of the very recent and
important result of \cite{maire-douc-olsson2013} which allows one
to compare performance of certain inhomogeneous Markov chains in the
case $k=2$. Our proof sheds some light on the developments in \cite{maire-douc-olsson2013}
and illustrates the difficulty encountered when trying to establish
the result for $k\geq3$ (Theorem \ref{thm:maire:douc:olsson}).

\section{Homogeneous embedding\label{sec:Homogeneous-embedding}}

The proofs of our results are based on classical Hilbert space techniques.
We recall here related definitions which will be useful throughout.
For any probability measure $\mu$ on $(\mathsf{E,\mathcal{E})}$
a family of Markov kernels $\left\{ \varPi_{i}:\mathsf{E}\times\mathcal{E}\rightarrow[0,1],i=1,\ldots,k\right\} $
and any function $f:\mathsf{E}\rightarrow\mathbb{R}$ let, whenever
the integrals are well-defined, 
\[
\mu\bigl(f\bigr)=\int f(x)\mu({\rm d}x)\qquad\text{and}\qquad\varPi_{q}f(x)=\int f(y)\varPi_{q}\bigl(x,{\rm d}y\bigr)\;,
\]
and for $k\geq2$, by induction, 
\[
\varPi_{\sigma^{0:k}(q)}f(x)=\int\varPi_{\sigma^{0:k-1}(q)}\bigl(x,{\rm d}y\bigr)\varPi_{\sigma^{k}(q)}f(y)\;.
\]
Consider next the spaces of square integrable (and \foreignlanguage{british}{centred})
functions defined respectively as 
\[
L^{2}(\mathsf{E},\mu)=\left\{ f:\mathsf{E}\rightarrow\mathbb{R\,}:\,\mu(f^{2})<\infty\right\} \qquad\text{and\qquad}L_{0}^{2}(\mathsf{E},\mu)=\big\{ f\in L^{2}(\mathsf{E},\mu)\,:\,\mu(f)=0\big\}\;,
\]
endowed with the inner product defined for any $f,g\in L^{2}(\mathsf{E},\mu)$
as $\left\langle f,g\right\rangle _{\mu}=\int f(x)g(x)\mu({\rm d}x)$,
and the associated norm $\|f\|_{\mu}=\sqrt{\langle f,f\bigr\rangle_{\mu}}$.
For $\lambda\in(0,1)$ and $f\in L^{2}(\mathsf{E},\mu)$ we introduce
the quantity 
\begin{align}
{\rm var}_{\lambda}\bigl(f,P^{{\rm strat}}\bigr) & =\|f-\mu\bigl(f\bigr)\|_{\mu}^{2}+\frac{2}{k}\sum_{q=1}^{k}\sum_{s=1}^{\infty}\lambda^{s}\left\langle f-\mu\bigl(f\bigr),\varPi_{\sigma^{0:s-1}(q)}f-\mu\bigl(f\bigr)\right\rangle _{\mu}\quad,\label{eq:asymptoticgeneralvariancestrat}
\end{align}
which, with an abuse of language, we may refer to as the asymptotic
variance. This quantity is well defined as for $f\in L^{2}(\mathsf{E},\mu)$
$\bigl|\bigl\langle f,\varPi_{\sigma^{0:s-1}(p)}f\bigr\rangle_{\mu}\bigr|\leq\|f\|_{\mu}^{2}<\infty$
for $s\in\mathbb{N}$ while the limit as $\lambda\uparrow1$ may or
may not exist. We first establish an expression for the asymptotic
variance of $M^{1/2}S_{M}(f)$, under minimal conditions, which can
be informally thought of as $\lim_{\lambda\uparrow1}{\rm var}_{\lambda}\bigl(f,P^{{\rm strat}}\bigr)$.
A similar expression was obtained in \cite{greenwood1998information}
for the Gibbs sampler and for $k=2$ in \cite{maire-douc-olsson2013}
under a slightly stronger assumption. 
\begin{prop}
\label{prop:expressioncycleasymptvar}Let $f\in L_{0}^{2}(\mathsf{X},\pi)$
and assume that for $q\in\{1,\ldots,k\}$ 
\[
\sum_{s=1}^{\infty}\left\langle f,\Pi_{\sigma^{0:s-1}(q)}f\right\rangle _{\pi}
\]
exists. Then for the inhomogeneous chain defined by $P^{{\rm strat}}$
\[
\lim_{M\rightarrow\infty}{\rm var}_{\pi}\left(M^{1/2}S_{M}\big(f\big)\right)=\|f\|_{\pi}^{2}+\frac{2}{k}\sum_{q=1}^{k}\sum_{s=1}^{\infty}\left\langle f,\Pi_{\sigma^{0:s-1}(q)}f\right\rangle _{\pi}\quad.
\]

\end{prop}
The proof can be found in the appendix. We now embed the inhomogeneous
Markov chain of the previous section into an homogeneous Markov chain,
which facilitates later analysis. This allows us to find a simple
expression for ${\rm var}_{\lambda}\bigl(f,P^{{\rm strat}}\bigr)$
(Lemma \ref{lem:resolventTandaVariance} and Corollary \ref{cor:asymptvarianceandT})
and our result is then a direct consequence of the standard result
in Lemma \ref{lem:inverseIminusPi-non-reversible}. Let $T:\mathsf{X}^{k}\times\mathcal{X}^{\varotimes k}\rightarrow[0,1]$
be the Markov transition probability such that
\begin{align*}
T\bigl(x^{(1)},\ldots,x^{(k)};{\rm d}y^{(1)}\times{\rm d}y^{(2)}\times\ldots\times{\rm d}y^{(k)}\bigr) & =\prod_{i=1}^{k}\Pi_{i}\bigl(x^{(i)},{\rm d}y^{(\sigma(i))}\bigr)
\end{align*}
and $\pi^{\varotimes k}\bigl({\rm d}x^{(1)}\times{\rm d}x^{(2)}\times\cdots\times{\rm d}x^{(k)}\bigr)=\prod_{i=1}^{k}\pi\bigl({\rm d}x^{(i)}\bigr)$.
We then define the associated time-homogeneous Markov chain $\bigl\{\bigl(X_{i}^{(1)},X_{i}^{(2)},\ldots,X_{i}^{(k)}\bigr),i\geq0\bigr\}$
such that $\bigl(X_{0}^{(1)},X_{0}^{(2)},\ldots,X_{0}^{(k)}\bigr)\sim\pi^{\varotimes k}$.
We note that the sequence$\bigl\{ X_{i}^{(\sigma^{i}(1))},i\geq0\bigr\}$
coincides with the non-homogeneous chain defined through $P^{{\rm strat}}$
in the introduction (a proof is provided in Lemma \ref{lem:linkinhomogeneousandT}),
and the other embedded chains correspond to the same cycle, but started
at different points of the cycle.  

As is customary we will use the same notation for the operator on
functions associated to the transition kernel $T$, that is letting
$L^{2,k}\bigl(\mathsf{X},\pi\bigr)=\bigl(L^{2}\bigl(\mathsf{X},\pi\bigr)\bigr)^{k}$
we define 
\begin{align*}
T:L^{2,k}\bigl(\mathsf{X},\pi\bigr) & \rightarrow L^{2,k}\bigl(\mathsf{X},\pi\bigr)
\end{align*}
such that for any $\varphi=\big(\varphi_{1},\varphi_{2},\ldots,\varphi_{k}\big)\in L^{2,k}\bigl(\mathsf{X},\pi\bigr)$
\[
T\varphi=\bigl(\Pi_{1}\varphi_{\sigma(1)},\Pi_{2}\varphi_{\sigma(2)},\ldots,\Pi_{i}\varphi_{\sigma(i)},\ldots,\Pi_{k}\varphi_{\sigma(k)}\bigr)\quad.
\]
We will let $T^{*}$ denote the adjoint of $T$ and endow the vector
space $L^{2,k}\bigl(\mathsf{X},\pi\bigr)$ with the inner product
defined for any $\varphi,\psi\in L^{2,k}\bigl(\mathsf{X},\pi\bigr)$
as 
\[
\left\langle \varphi,\psi\right\rangle =\sum_{i=1}^{k}\bigl\langle\varphi_{i},\psi_{i}\bigr\rangle_{\pi}\quad.
\]
The following result allows us to write (\ref{eq:asymptoticgeneralvariancestrat})
in terms of the resolvent of $T$. We use the standard convention
that $T^{0}=I$, the identity operator, the additional convention
that for $j\in\mathbb{N}$, $\Pi_{\sigma^{0:-1}(j)}=I$ and define
for $\lambda\in(0,1)$ and $\varphi\in L^{2,k}\bigl(\mathsf{X},\pi\bigr)$,
$\bigl(I-\lambda T\bigr)^{-1}\varphi=\sum_{i=0}^{\infty}\lambda^{i}T^{i}\varphi$.
\begin{lem}
\label{lem:resolventTandaVariance}Let $f\in L^{2}\bigl(\mathsf{X},\pi\bigr)$
and define $\bar{f}\in\left\{ f\right\} ^{k}\in L^{2,k}\bigl(\mathsf{X},\pi\bigr)$.
Then for any $\lambda\in(0,1)$ we have
\begin{align*}
\bigl\langle\bar{f},\bigl(I-\lambda T\bigr)^{-1}\bar{f}\bigr\rangle & =\sum_{i=0}^{\infty}\lambda^{i}\sum_{q=1}^{k}\bigl\langle f,\Pi_{\sigma^{0:i-1}(q)}f\bigr\rangle=\sum_{q=1}^{k}\sum_{i=0}^{\infty}\lambda^{i}\bigl\langle f,\Pi_{\sigma^{0:i-1}(q)}f\bigr\rangle\;.
\end{align*}
\end{lem}
\begin{proof}
Let $\varphi\in L^{2,k}\bigl(\mathsf{X},\pi\bigr)$. We first establish
that for any $i\geq1$ we have for all $j\in\{1,\ldots,k\}$ 
\[
\bigl[T^{i}\varphi\bigr]_{j}=\Pi_{j}\Pi_{\sigma(j)}\cdots\Pi_{\sigma^{i-1}(j)}\varphi_{\sigma^{i}(j)}\;.
\]
This is clearly true for $i=1$. Assume this is true for $i\geq1$
then 
\[
\bigl[T^{i+1}\varphi\bigr]_{j}=\bigl[T^{i}\circ T\varphi\bigr]_{j}=\Pi_{j}\Pi_{\sigma(j)}\cdots\Pi_{\sigma^{i-1}(j)}\bigl[T\varphi\bigr]_{\sigma^{i}(j)}=\Pi_{j}\Pi_{\sigma(j)}\cdots\Pi_{\sigma^{i-1}(j)}\Pi{}_{\sigma^{i}(j)}\varphi_{\sigma^{i+1}(j)}\;,
\]
from which we conclude. This implies that for $i\geq1$ 
\[
\bigl\langle\bar{f},T^{i}\bar{f}\bigr\rangle=\sum_{q=1}^{k}\bigl\langle f,\Pi_{\sigma^{0:i-1}(q)}f\bigr\rangle_{\pi}\;.
\]
Now with our conventions
\[
\bigl\langle\bar{f},\bigl(I-\lambda T\bigr)^{-1}\bar{f}\bigr\rangle=\sum_{i=0}^{\infty}\lambda^{i}\sum_{q=1}^{k}\bigl\langle f,\Pi_{\sigma^{0:i-1}(q)}f\bigr\rangle_{\pi}\;,
\]
where we note that the sum is absolutely convergent.\end{proof}
\begin{cor}
\label{cor:asymptvarianceandT}For any $f\in L_{0}^{2}\bigl(\mathsf{X},\pi\bigr)$
one can rewrite the asymptotic variance as follows
\begin{equation}
{\rm var}_{\lambda}\bigl(f,P^{{\rm strat}}\bigr)=\frac{2}{k}\bigl\langle\bar{f},\bigl(I-\lambda T\bigr)^{-1}\bar{f}\bigr\rangle-\|f\|_{\pi}^{2}\;,\label{eq:generalisedresolventexpressionvariance}
\end{equation}
which generalizes the expression for the asymptotic variance in the
homogeneous case ($k=1$) in terms of the resolvent.
\end{cor}
We have the following general result (which can be traced back at
least to \cite[proof of Lemma 3.1]{landim2004superdiffusivity}) which
leads to a powerful variational representation of the asymptotic variance
associated to general Markov transition probabilities. A proof is
provided in the supplementary material for completeness.
\begin{lem}
\label{lem:inverseIminusPi-non-reversible}Let $\big(\mathsf{E},\mathcal{E}\big)$
be a measurable space on which we define a probability distribution
$\mu$ and a Markov transition probability $\varPi:\mathsf{E}\times\mathcal{E}\to[0,1]$,
not necessarily reversible, leaving $\mu$ invariant. Then, with $S=\big(\varPi+\varPi^{*}\big)/2$
and $A=\big(\varPi-\varPi^{*}\big)/2$ (the self-adjoint and skew
symmetric parts of $\varPi$), for any $\lambda\in(0,1)$ and $f\in L^{2}\bigl(\mathsf{E},\mu\bigr)$
\[
\bigl\langle f,\big(I-\lambda\varPi\big)^{-1}f\bigr\rangle_{\mu}=\sup_{g\in L^{2}(\mathsf{E},\mu)}2\bigl\langle f,g\bigr\rangle_{\mu}-\bigl\langle g,\big(I-\lambda S\big)g\bigr\rangle_{\mu}-\lambda^{2}\bigl\langle Ag,\big(I-\lambda S\big)^{-1}Ag\bigr\rangle_{\mu}\quad.
\]
\end{lem}
\begin{cor}
\label{cor:inverseIminusPi-non-reversible}As a consequence
\[
\bigl\langle f,\big(I-\lambda\varPi\big)^{-1}f\bigr\rangle_{\mu}\leq\bigl\langle f,\big(I-\lambda S\big)^{-1}f\bigr\rangle_{\mu}-\lambda^{2}\bigl\langle A\hat{g},\big(I-\lambda S\big)^{-1}A\hat{g}\bigr\rangle_{\mu}\leq\bigl\langle f,\big(I-\lambda S\big)^{-1}f\bigr\rangle_{\mu}
\]
where $\hat{g}=\big(I-\lambda\varPi^{*}\big)^{-1}\big(I-\lambda S\big)\big(I-\lambda\varPi\big)^{-1}f$.
\end{cor}
Now we consider a direct application of this result which leads to
our main result, Theorem \ref{thm:mainresult}.
\begin{thm}
\label{thm:mainresult}Let $k=2$. For any \textup{$f\in L_{0}^{2}\bigl(\mathsf{X},\pi\bigr)$
and $\lambda\in[0,1)$} 
\[
\mathrm{var}_{\lambda}(f,P^{{\rm rand}})\geq\mathrm{var}_{\lambda}(f,P^{{\rm strat}})\;.
\]
If in addition $f\in L_{0}^{2}\bigl(\mathsf{X},\pi\bigr)$ satisfies
\begin{equation}
\sum_{q=1}^{k}\sum_{s=1}^{\infty}\bigl|\bigl\langle f,\Pi_{\sigma^{0:s-1}(q)}f\bigr\rangle_{\pi}\bigr|<\infty\quad,\label{eq:uglyassumption}
\end{equation}
then 
\[
\mathrm{var}(f,P^{{\rm rand}})\geq\mathrm{var}(f,P^{{\rm strat}})\;.
\]
\end{thm}
\begin{proof}
We let $S=\bigl(T+T^{*}\bigr)/2$ and $A=\bigl(T-T^{*}\bigr)/2$ be
the self-adjoint and skew symmetric parts of $T$. Notice that
\[
T\varphi=\bigl(\Pi_{1}\varphi_{2},\Pi_{2}\varphi_{1}\bigr)\mbox{\;\text{and}}\; T^{*}\varphi=\bigl(\Pi_{2}\varphi_{2},\Pi_{1}\varphi_{1}\bigr)
\]
where, if needed, the second statement can be established using Lemma
\ref{lem:properties of T}. Therefore 
\[
\frac{\big(T+T^{*}\big)}{2}\varphi=\left(\frac{\big(\Pi_{1}+\Pi_{2}\big)}{2}\varphi_{2},\frac{\big(\Pi_{1}+\Pi_{2}\big)}{2}\varphi_{1}\right)\;,
\]
which corresponds to two homogeneous chains run in parallel, each
with transition probability $P^{{\rm rand}}$. We now apply Corollary
\ref{cor:inverseIminusPi-non-reversible} to $T$ and obtain
\[
\bigl\langle\bar{f},\bigl(I-\lambda T\bigr)^{-1}\bar{f}\bigr\rangle-\|f\|_{\pi}^{2}\leq\bigl\langle\bar{f},\big(I-\lambda S\big)^{-1}\bar{f}\bigr\rangle-\|f\|_{\pi}^{2}
\]
and remark that for any $g\in L^{2,k}\bigl(\mathsf{X},\pi\bigr)$
\[
\bigl\langle\bar{f},\big(I-\lambda S\big)^{-1}\bar{f}\bigr\rangle=\sup_{g\in L^{2,k}(\mathsf{X},\pi)}2\bigl\langle\bar{f},g\bigr\rangle-\bigl\langle g,\big(I-\lambda S\big)g\bigr\rangle=2\bigl\langle f,\big(I-\lambda P^{{\rm rand}}\big)^{-1}f\bigr\rangle_{\pi}
\]
where we have used that $2\bigl\langle\bar{f},g\bigr\rangle-\bigl\langle g,\big(I-\lambda S\big)g\bigr\rangle=\sum_{i=1}^{2}\bigl\langle f,g_{i}\bigr\rangle_{\pi}-\bigl\langle g_{i},\big(I-\lambda P^{{\rm rand}}\big)g_{i}\bigr\rangle_{\pi}$
and Lemma \ref{lem:variationalrepinverseoperator} twice to establish
the two equalities. Therefore
\[
\bigl\langle\bar{f},\bigl(I-\lambda T\bigr)^{-1}\bar{f}\bigr\rangle-\|f\|_{\pi}^{2}\leq2\bigl\langle f,\big(I-\lambda P^{{\rm rand}}\big)^{-1}f\bigr\rangle_{\pi}-\|f\|_{\pi}^{2}
\]
and the first statement follows with Corollary \ref{cor:asymptvarianceandT}.
For the second statement, since $P^{{\rm rand}}$ is self-adjoint,
$\lim_{\lambda\uparrow1}\mathrm{var}_{\lambda}(f,P^{{\rm rand}})$
exists and converges to $\mathrm{var}(f,P^{{\rm rand}})$ (if finite)
and the additional summability condition allows us to conclude in
a similar fashion that $\lim_{\lambda\uparrow1}\mathrm{var}_{\lambda}(f,P^{{\rm strat}})=\mathrm{var}(f,P^{{\rm strat}})$.
\end{proof}
For $k=2$, cycling deterministically through $\mathfrak{P}$ is therefore
always better in terms of asymptotic variance than random scanning.
A related result was previously known for the Gibbs sampler \cite{greenwood1998information}
that is in the particular scenario where both $\Pi_{1}$ and $\Pi_{2}$
are projections, a property essential in order to establish that
\[
\mathrm{var}(f,P^{{\rm rand}})=2\mathrm{var}(f,P^{{\rm strat}})-{\rm var}_{\pi}\big(f\big)\;.
\]
In the present scenario from Corollary \ref{cor:inverseIminusPi-non-reversible}
one can obtain a lower bound on the gap in the inequality, $\bigl\langle A\hat{g},\big(I-\lambda S\big)^{-1}A\hat{g}\bigr\rangle$
where 
\[
\hat{g}=\big(I-\lambda T^{*}\big)^{-1}\big(I-\lambda S\big)\big(I-\lambda T\big)^{-1}\bar{f}\;.
\]

\section{A short proof of the ordering result of \cite{maire-douc-olsson2013}\label{sec:A-short-proof}}

In \cite{maire-douc-olsson2013} the authors have established that
for $k=2$ the algorithm satisfies a Peskun type result \cite{peskun,tierney-note}.

\begin{thm}[Maire, Douc and Olsson]
\label{thm:maire:douc:olsson}Let $k=2$ and consider two pairs of
$\pi-$reversible Markov transition probabilities $\mathfrak{P}=\big\{\Pi_{1},\Pi_{2}\big\}$
and $\check{\mathfrak{P}}=\big\{\check{\Pi}_{1},\check{\Pi}_{2}\big\}$.
If for any $g\in L^{2}\big(\mathsf{E},\pi\big)$ and $i\in\{1,2\}$
$\bigl\langle g,(I-\check{\Pi}_{i})g\bigr\rangle_{\pi}\leq\bigl\langle g,(I-\Pi_{i})g\bigr\rangle_{\pi}$
then for any $f\in L^{2}\big(\mathsf{X},\pi\big)$, ${\rm var}_{\lambda}\bigl(f,P^{{\rm strat}}\bigr)\leq{\rm var}_{\lambda}\bigl(f,\check{P}^{{\rm strat}}\bigr)$.
If in addition (\ref{eq:uglyassumption}) holds for $\mathfrak{P}$
and $\mathfrak{P}'$ then ${\rm var}\bigl(f,P^{{\rm strat}}\bigr)\leq{\rm var}\bigl(f,\check{P}^{{\rm strat}}\bigr)$.
\end{thm}
The proof of this fact is a direct consequence of Lemma \ref{lem:derivativeofIminusT}
below. In order to state this key result it is useful to rewrite the
operator $T$ as the composition of elementary operators, 
\begin{align*}
\mathfrak{S}^{\pm1},\Delta\colon L^{2,k}\bigl(\mathsf{X},\pi\bigr) & \rightarrow L^{2,k}\bigl(\mathsf{X},\pi\bigr)
\end{align*}
where $\mathfrak{S}^{\pm1}$ are the forward and backward circular
permutation operators such that for any $\varphi\in L^{2,k}\bigl(\mathsf{X},\pi\bigr)$
\begin{align*}
\mathfrak{S}^{\pm1}\varphi & =\bigl(\varphi_{\sigma^{\pm1}(1)},\varphi_{\sigma^{\pm1}(2)},\ldots,\varphi_{\sigma^{\pm1}(i)},\ldots,\varphi_{\sigma^{\pm1}(k)}\big)
\end{align*}
and $\Delta$ is the diagonal operator such that for any $\varphi\in L^{2,k}\bigl(\mathsf{X},\pi\bigr)$
\begin{align*}
\Delta\varphi & =\bigl(\Pi_{1}\varphi_{1},\ldots,\Pi_{k-1}\varphi_{k-1},\Pi_{k}\varphi_{k}\bigr)\;.
\end{align*}
Then $T=\Delta\circ\mathfrak{S}$ (see Lemma \ref{lem:properties of T}).
\begin{lem}
\label{lem:derivativeofIminusT}Let $T$ and $\check{T}$ be the embedding
operators as defined earlier and associated with $\mathfrak{P}$ and
$\check{\mathfrak{P}}$ given in Theorem \ref{thm:maire:douc:olsson},
define $T(\beta)=\beta T+(1-\beta)\check{T}$ and $\Delta\big(\beta\big)=\beta\Delta+(1-\beta)\check{\Delta}$
for $\beta\in[0,1]$. For $\lambda\in[0,1)$ and $f\in L^{2}\big(\mathsf{X},\pi\big)$
let $\delta_{\lambda,f}(\beta)=\bigl\langle\bar{f},\big(I-\lambda T(\beta)\big)^{-1}\bar{f}\bigr\rangle$,
with $\bar{f}\in\{f\}^{k}$. Then
\[
\frac{\partial}{\partial\beta}\delta_{\lambda,f}(\beta)=\lambda\bigl\langle\big(I-\mathfrak{S}^{-1}\Delta(\beta)\big)^{-1}\bar{f},\big(\Delta-\check{\Delta}\big)\big(I-\mathfrak{S}\Delta(\beta)\big)^{-1}\bar{f}\bigr\rangle\quad.
\]
\end{lem}
\begin{proof}
We note that $T(\beta)=\Delta\big(\beta\big)\circ\mathfrak{S}$ and
have 
\begin{align*}
\frac{\partial}{\partial\beta}\delta{}_{\lambda,f}(\beta) & =\lambda\bigl\langle\bar{f},\big(I-\lambda T(\beta)\big)^{-1}\big(T-\check{T}\big)\big(I-\lambda T(\beta)\big)^{-1}\bar{f}\bigr\rangle\\
 & =\lambda\bigl\langle\bar{f},\big(I-\lambda T(\beta)\big)^{-1}\big(\Delta-\check{\Delta}\big)\mathfrak{S}\big(I-\lambda T(\beta)\big)^{-1}\mathfrak{S}^{-1}\bar{f}\bigr\rangle\\
 & =\lambda\bigl\langle\big(I-\lambda\mathfrak{S}^{-1}T(\beta)\mathfrak{S}^{-1}\big)^{-1}\bar{f},\big(\Delta-\check{\Delta}\big)\mathfrak{S}\big(I-\lambda T(\beta)\big)^{-1}\mathfrak{S}^{-1}\bar{f}\bigr\rangle\\
 & =\lambda\bigl\langle\big(I-\lambda\mathfrak{S}^{-1}T(\beta)\mathfrak{S}^{-1}\big)^{-1}\bar{f},\big(\Delta-\check{\Delta}\big)\big(I-\lambda\mathfrak{S}T(\beta)\mathfrak{S}^{-1}\big)^{-1}\bar{f}\bigr\rangle\,.
\end{align*}
The first equality follows from standard arguments (see \cite{tierney-note}
and \cite[Lemma 51]{andrieu-vihola-2012} for additional details,
noting that reversibility is in fact not required). On the second
line we have used the definition of $T$ and $\check{T}$ in terms
of $\Delta,\check{\Delta}$ and $\mathfrak{S}$, and $\mathfrak{S}^{\pm1}\bar{f}=\bar{f}$.
On the third line we have used that for $k\geq1$ the adjoint of $T(\beta)^{k}$
is $T^{*}(\beta)^{k}=\big(\mathfrak{S}^{-1}T(\beta)\mathfrak{S}^{-1}\big)^{k}$
from Lemma \ref{lem:properties of T}, from which we deduce that $\big[\big(I-\lambda T(\beta)\big)^{-1}\big]^{*}=\big(I-\lambda\mathfrak{S}^{-1}T(\beta)\mathfrak{S}^{-1}\big)^{-1}$.
On the fourth line we use the identity
\begin{multline*}
\mathfrak{S}\big(I-\lambda T(\beta)\big)\mathfrak{S}^{-1}=\mathfrak{S}\left(\sum_{k=0}^{\infty}\big(\lambda T(\beta)\big)^{k}\right)\mathfrak{S}^{-1}=\sum_{k=0}^{\infty}\big(\lambda\mathfrak{S}T(\beta)\mathfrak{S}^{-1}\big)^{k}=\big(I-\lambda\mathfrak{S}T(\beta)\mathfrak{S}^{-1}\big)^{-1}
\end{multline*}
and the result follows with $T\big(\beta\big)=\Delta\big(\beta\big)\mathfrak{S}$.
\end{proof}

\begin{proof}[Proof of Theorem \ref{thm:maire:douc:olsson}]
The derivative of $\delta_{\lambda,f}(\beta)$ is evidently positive
if the resolvent type terms $\big(I-\mathfrak{S}\Delta\big(\beta\big)\big)^{-1}\bar{f}$
and $\big(I-\mathfrak{S}^{-1}\Delta\big(\beta\big)\big)^{-1}\bar{f}$
coincide since for any $\phi\in L^{2,k}\bigl(\mathsf{X},\pi\bigr)$
$\bigl\langle\phi,\big(\Delta-\check{\Delta}\big)\phi\bigr\rangle\geq0$.
This is the case for $k=2$ since in this scenario $\mathfrak{S}^{-1}=\mathfrak{S}$
which reduces to a ``swap'' operator.\end{proof}
\begin{rem}
We may wonder whether the result of Theorem \ref{thm:maire:douc:olsson}
holds true for $k\geq3$. To that purpose assume that the Markov transitions
in $\mathfrak{P}$ and $\check{\mathfrak{P}}$ coincide, except for
the $i-$th element and notice that with $\Pi_{j}(\beta)=\big[\Delta(\beta)\big]_{j}$
\[
\left[\big(I-\mathfrak{S}\Delta\big(\beta\big)\big)^{-1}\bar{f}\right]_{i}=\sum_{j=0}^{\infty}\lambda^{j}\Pi_{\sigma^{1:j}(i)}\big(\beta\big)f
\]
and
\[
\left[\big(I-\lambda\mathfrak{S}^{-1}\Delta\big(\beta\big)\big)^{-1}\bar{f}\right]_{i}=\sum_{j=0}^{\infty}\lambda^{j}\Pi_{\sigma^{-j:-1}(i)}\big(\beta\big)f\;.
\]
Then the derivative of $\delta{}_{\lambda,f}(\beta)$ reduces to 
\[
\frac{\partial}{\partial\beta}\delta{}_{\lambda,f}(\beta)=\bigl\langle\sum_{j=0}^{\infty}\lambda^{j}\Pi_{\sigma^{-j:-1}(i)}\big(\beta\big)f,(\Pi_{i}-\check{\Pi}{}_{i})\sum_{j=0}^{\infty}\lambda^{j}\Pi_{\sigma^{1:j}(i)}\big(\beta\big)f\bigr\rangle_{\pi}
\]
and this term is positive as soon as the two sums coincide, which
is the case if $k=2p-1$ for some $p\in\mathbb{N}^{*}$, $i=p$ and
$(\Pi_{1},\Pi_{2},\ldots,\Pi_{k})=(Q_{p},Q_{p-1},\ldots,Q_{2},Q_{1},Q_{2},\ldots,Q_{p})$
for a family of transition probabilities $\big\{ Q_{i}:\mathsf{X}\times\mathcal{X}\rightarrow[0,1],i=1,\ldots,p\big\}$.
This corresponds to a known reverbilisation strategy of the Gibbs
sampler, albeit for the operator $Q_{1}Q_{2}\cdots Q_{p}$. We note
that this also holds in the case where $k=2p-2$ for some $p\in\mathbb{N}^{*}$,
the cycle $(\Pi_{1},\Pi_{2},\ldots,\Pi_{k})=(Q_{p-1},\ldots,Q_{2},Q_{1},Q_{2},\ldots,Q_{p})$
and both $i=p-1$ and $i=2p-2$.
\end{rem}

\section{Conclusions and perspectives}

We have introduced a novel time-homogeneous Markov embedding of a
class of time inhomogeneous Markov chains widely used in the context
of Monte Carlo sampling techniques. We have shown that this approach
allows one to rapidly prove new, or recently established, results
by leveraging existing techniques or known results for homogeneous
Markov chains. We suggest two possible directions for further research.
In \cite{toth:1986} the author extended the celebrated Kipnis and
Varadhan results on the central limit theorem for reversible Markov
chains to the general scenario; our approach offers the promise to
be able to extend those results to the inhomogeneous Markov chains
considered in the present paper. Another possible avenue of research
is concerned with finding bounds on the convergence to stationarity
in terms of e.g. total variation distance. For example one could attempt
to extend the results of Fill \cite[Theorem 2.1]{fill1991} which
rely on the multiplicative \foreignlanguage{british}{reverbilisation}
of non-reversible Markov chains and the assumption that $\mathsf{X}$
is a finite discrete space.

\appendix

\section{Appendix}
\begin{proof}[Proof of Proposition \ref{prop:expressioncycleasymptvar}]
Let $f\in L_{0}^{2}(\mathsf{X},\pi)$ then for $n\geq1$ 
\[
\mathbb{E}_{\pi}\bigl[{\textstyle \sum}_{i=0}^{n-1}f\bigl(X_{i}\bigr){\textstyle \sum}_{j=0}^{n-1}f\bigl(X_{j}\bigr)\bigr]=n\mathbb{E}_{\pi}\left[f^{2}(X_{0})\right]+2\sum_{0\leq i<j\leq n-1}\mathbb{E}_{\pi}\bigl[f\bigl(X_{i}\bigr)f\bigl(X_{j}\bigr)\bigr]
\]
and we focus on the second term. We rewrite it as
\begin{align*}
\sum_{0\leq i<j\leq n-1}\mathbb{E}_{\pi}\bigl[f\bigl(X_{i}\bigr)f\bigl(X_{j}\bigr)\bigr] & =\sum_{0\leq i<j\leq n-1}\mathbb{E}_{\pi}\bigl[f\bigl(X_{i}\bigr)\Pi_{\sigma^{i:j-1}(1)}f\bigl(X_{i}\bigr)\bigr]\\
 & =\sum_{0\leq i<j\leq n-1}\bigl\langle f,\Pi_{\sigma^{i:j-1}(1)}f\bigr\rangle_{\pi}\\
 & =\sum_{0\leq i<j\leq n-1}\bigl\langle f,\Pi_{\sigma^{i:j-i-1+i}(1)}f\bigr\rangle_{\pi}\\
 & =\sum_{0\leq i<n-1}\sum_{m=1}^{n-1-i}\bigl\langle f,\Pi_{\sigma^{i:m-1+i}(1)}f\bigr\rangle_{\pi}\quad.
\end{align*}
Now we have (the term $i=0$ is treated similarly, but separately)
\begin{multline*}
n^{-1}\sum_{0<i<n-1}\sum_{m=1}^{n-1-i}\bigl\langle f,\Pi_{\sigma^{i:m-1+i}(1)}f\bigr\rangle_{\pi}\\
=\sum_{q=1}^{k}\frac{\bigl\lfloor(n-2-q)/k\bigr\rfloor}{n}\frac{1}{\bigl\lfloor(n-2-q)/k\bigr\rfloor}\sum_{0<pk+q<n-1}\sum_{m=1}^{n-1-pk-q}\bigl\langle f,\Pi_{\sigma^{0:m-1}(q)}f\bigr\rangle_{\pi}\;,
\end{multline*}
and we conclude by letting $n\rightarrow\infty$ and using a Cesàro
sum argument for each $q\in\{1,\ldots,k\}$.
\end{proof}
Operator $T$ is not self-adjoint, but one can easily determine the
expression for its adjoint $T^{*}$ in terms of $\mathfrak{S}$ and
$\Delta$ or $T$ (\foreignlanguage{british}{visualising} $T$ as
a block diagonal matrix may be helpful). 
\begin{lem}
\label{lem:properties of T}We have that \end{lem}
\begin{enumerate}
\item the adjoint of $\mathfrak{S}$ is $\mathfrak{S}^{*}=\mathfrak{S}^{-1}$,
\item $\Delta^{*}=\Delta$, that is $\Delta$ is self-adjoint,
\item $T=\Delta\circ\mathfrak{S}$ and the adjoint of $T$ is $T^{*}=\mathfrak{S}^{-1}\circ\Delta=\mathfrak{S}^{-1}\circ T\circ\mathfrak{S}^{-1}$.\end{enumerate}
\begin{proof}
The first result follows from 
\begin{align*}
\bigl\langle\varphi,\mathfrak{S}\varphi\bigr\rangle & =\sum_{i=1}^{k}\bigl\langle\varphi_{i},\varphi_{\sigma(i)}\bigr\rangle_{\pi}=\sum_{j=1}^{k}\bigl\langle\varphi_{\sigma^{-1}(j)},\varphi_{j}\bigr\rangle_{\pi}=\bigl\langle\mathfrak{S}^{-1}\varphi,\varphi\bigr\rangle\;.
\end{align*}
The second result is direct and the third result follows from the
general fact that $T^{*}=\mathfrak{S}^{*}\circ\Delta^{*}$ followed
by an application of the first two results of the lemma. We conclude.
\end{proof}

\section{Supplementary material}

We let $\bar{\mathbb{P}}(\cdot)$ denote the probability distribution
of the Markov chain defined by $T$.
\begin{lem}
\label{lem:linkinhomogeneousandT}For $m\geq0$ and $A_{1},A_{2},\ldots,A_{m}\in\mathcal{X}^{m}$
we have 
\[
\bar{\mathbb{P}}\left(X_{0}^{\sigma^{0}(1)}\in A_{0},X_{1}^{\sigma^{1}(1)}\in A_{1},\ldots,X_{m}^{\sigma^{m}(1)}\in A_{m}\right)=\mathbb{P}\left(X_{0}\in A_{0},X_{1}\in A_{1},\ldots,X_{m}\in A_{m}\right)\;.
\]
\end{lem}
\begin{proof}
By construction for $A\in\mathcal{X}$, $\bar{\mathbb{P}}\left(X_{0}^{\sigma^{0}(1)}\in A\right)=\pi\bigl(A\bigr)=\mathbb{P}\left(X_{0}\in A\right)$
and for $i\geq1$ component $\sigma^{i}(1)$ is generated by kernel
$\Pi_{\sigma^{i-1}(1)}$ and with 
\[
\mathcal{G}_{i}=\sigma\left(\bigl(X_{m}^{(1)},X_{m}^{(2)},\ldots,X_{m}^{(k)}\bigr),0\leq m\leq i\right)
\]
 we have 
\[
\bar{\mathbb{P}}\left(X_{i}^{\sigma^{i}(1)}\in A\mid\mathcal{G}_{i-1}\right)=\Pi_{\sigma^{i-1}(1)}\left(X_{i-1}^{\sigma^{i-1}(1)},A\right)
\]
and we conclude.
\end{proof}

\begin{proof}[Proof of Lemma \ref{lem:inverseIminusPi-non-reversible}]
We have 
\begin{align*}
\bigl\langle f,\bigl[I-\lambda\varPi\bigr]^{-1}f\bigr\rangle_{\mu} & =\bigl\langle\big(I-\lambda\varPi\big)\big(I-\lambda\varPi\big)^{-1}f,\big(I-\lambda\varPi\big)^{-1}f\bigr\rangle_{\mu}\\
 & =\bigl\langle\big(I-\lambda S\big)\big(I-\lambda\varPi\big)^{-1}f,\big(I-\lambda\varPi\big)^{-1}f\bigr\rangle_{\mu}\\
 & =\sup_{h\in L^{2}(\mathsf{E},\mu)}2\bigl\langle\big(I-\lambda\varPi\big)^{-1}f,h\bigr\rangle_{\mu}-\bigl\langle h,\big(I-\lambda S\big)^{-1}h\bigr\rangle_{\mu}\\
 & =\sup_{h\in L^{2}(\mathsf{E},\mu)}2\bigl\langle f,\big(I-\lambda\varPi^{*}\big)^{-1}h\bigr\rangle_{\mu}-\bigl\langle h,\big(I-\lambda S\big)^{-1}h\bigr\rangle_{\mu}\\
 & =\sup_{g\in L^{2}(\mathsf{E},\mu)}2\bigl\langle f,g\bigr\rangle_{\mu}-\bigl\langle\big(I-\lambda\varPi^{*}\big)g,\big(I-\lambda S\big)^{-1}\big(I-\lambda\varPi^{*}\big)g\bigr\rangle_{\mu}\\
 & =\sup_{g\in L^{2}(\mathsf{E},\mu)}2\bigl\langle f,g\bigr\rangle_{\mu}-\bigl\langle g,\big(I-\lambda S\big)g\bigr\rangle_{\mu}-\bigl\langle\lambda Ag,\big(I-\lambda S\big)^{-1}\lambda Ag\bigr\rangle_{\mu}
\end{align*}
where we have used that $\varPi=S+A$, $\varPi^{*}=S-A$, that for
any $g\in L^{2}\bigl(\mathsf{E},\mu\bigr)$, $\bigl\langle g,Ag\bigr\rangle_{\mu}=-\bigl\langle Ag,g\bigr\rangle_{\mu}=0$,
Lemma \ref{lem:variationalrepinverseoperator} for the self-adjoint
operator $\big(I-\lambda S\big)$, $\big[\big(I-\lambda\varPi\big)^{-1}\big]^{*}=\big(I-\lambda\varPi^{*}\big)^{-1}$,
set $g=\big(I-\lambda\varPi^{*}\big)^{-1}h$ and again used the property
$\bigl\langle g,Ag\bigr\rangle_{\mu}=0$. From Lemma \ref{lem:variationalrepinverseoperator}
the supremum on the third line is attained for $\hat{h}=\big(I-\lambda S\big)\big(I-\lambda\varPi\big)^{-1}f$,
which translates into $\hat{g}=\big(I-\lambda\varPi^{*}\big)^{-1}\hat{h}$
on the last line. Consequently using again Lemma \ref{lem:variationalrepinverseoperator}
for the operator $I-\lambda S$ we deduce 
\[
\bigl\langle f,\bigl[I-\lambda\varPi\bigr]^{-1}f\bigr\rangle_{\mu}\leq\bigl\langle f,\bigl[I-\lambda S\bigr]^{-1}f\bigr\rangle_{\mu}-\lambda^{2}\bigl\langle A\hat{g},\big(I-\lambda S\big)^{-1}A\hat{g}\bigr\rangle_{\mu}
\]

\end{proof}
The following provides a useful variational representation of the
quadratic form of the inverse of a positive self-adjoint operators,
attributed to Bellman, and used for example by \cite{caracciolo-pelissetto-sokal}. 
\begin{lem}
\label{lem:variationalrepinverseoperator}Let $A$ be a self-adjoint
operator on a Hilbert space $\mathcal{H}$, satisfying $\left\langle f,Af\right\rangle \ge0$
for all $f\in\mathcal{H}$ and such that the inverse $A^{-1}$ exists.
Then 
\[
\left\langle f,A^{-1}f\right\rangle =\sup_{g\in\mathcal{H}}\left[2\left\langle f,g\right\rangle -\left\langle g,Ag\right\rangle \right]\;,
\]
 where the supremum is attained with $g=A^{-1}f$. 
\end{lem}

\appendix

\bibliographystyle{plain}
\bibliography{refs}

\end{document}